\documentclass{svproc}
\usepackage{url}
\usepackage{hyperref}
\hypersetup{
    colorlinks=true,
    linkcolor=blue,
    filecolor=magenta,      
    urlcolor=cyan,
}
\usepackage{graphicx}

\usepackage{amsmath}
\usepackage{amsfonts}
\usepackage{booktabs}
\usepackage[ruled,linesnumbered]{algorithm2e}
\usepackage{tcolorbox}

\begin{document}
\mainmatter              
\title{Influential Billboard Slot Selection using Pruned Submodularity Graph \thanks{The work of Dr. Suman Banerjee is supported by the Start Up Research Grant provided by Indian Institute of Technology Jammu, India (Grant No.: SG100047).} }
\titlerunning{Influential Billboard Slot Selection using Pruned Submodularity Graph}  
%
\author{Dildar Ali \and Suman Banerjee \and Yamuna Prasad}
\authorrunning{Ali et al.} 
%
\tocauthor{Ivar Ekeland, Roger Temam, Jeffrey Dean, David Grove,
Craig Chambers, Kim B. Bruce, and Elisa Bertino}
\institute{Department of Computer Science and Engineering,\\
Indian Institute of Technology Jammu \\
Jammu \& Kashmir-181221, India.\\
\email{2021rcs2009@iitjammu.ac.in, suman.banerjee@iitjammu.ac.in, yamuna.prasad@iitjammu.ac.in}}

\maketitle              

\begin{abstract}
Billboard Advertisement has emerged as an effective out-of-home advertisement technique and adopted by many commercial houses. In this case, the billboards are owned by some companies  and they are provided to the commercial houses slot\mbox{-}wise on a payment basis. Now, given the database of billboards along with their slot information which $k$ slots should be chosen to maximize the influence. Formally, we call this problem as the \textsc{Influential Billboard Slot Selection} Problem. In this paper, we pose this problem as a combinatorial optimization problem. Under the `triggering model of influence', the influence function is non-negative, monotone, and submodular. However, as the incremental greedy approach for submodular function maximization does not scale well along with the size of the problem instances, there is a need to develop efficient solution methodologies for this problem. 
\par In this paper, we apply the pruned submodularity graph-based pruning technique for solving this problem. The proposed approach is divided into three phases, namely, preprocessing, pruning, and selection. We analyze the proposed solution approach for its performance guarantee and computational complexity. We conduct an extensive set of experiments with real-world datasets and compare the performance of the proposed solution approach with many baseline methods. We observe that the proposed one leads to more amount of influence compared to all the baseline methods within reasonable computational time.


\end{abstract}
\section{Introduction}
Creating and maximizing influence among the customers is one of the main goals of a commercial house. For this purpose, they spend around $7-10 \%$ of their annual revenue. Now, how to make use of this budget effectively remains an active area of research. There are several ways through which advertisement can be done such as social media, television channels, and many more. However, billboard advertisement has emerged as an effective approach for out-of-home advertisement as it provides more return on investment compared to other advertisement techniques. In recent times the billboards are digital and they are allocated time slot\mbox{-}wise to the commercial houses based on their payments.


\paragraph{\textbf{Problem Background}} In this advertisement technique, often different commercial houses select a limited number of billboard slots with the hope that the advertisement content will be observed by many people and a significant number of them will be influenced towards the product. This may increase the sales and revenue earned from the product. Now, due to the budget constraint, only a very small number of billboard slots can be affordable for the E-Commerce house. So, it's a prominent question that for a given value of $k \in \mathbb{Z}^{+}$, which of the billboard slots should be chosen for posting the advertisement content. recently, this problem has been studied by many researchers and different kinds of solution approaches have been proposed \cite{wang2022data}.

\paragraph{\textbf{Problem Definition}} Now-a-days due to the advancement of wireless devices and mobile internet, capturing the location information of moving objects become easier. This leads to the availability of many trajectory datasets in different repositories and they are being used to solve many real-life problems including route recommendation \cite{dai2015personalized,qu2019profitable}, driving behavior prediction \cite{xue2019rapid}, and many more. As mentioned previously, these trajectory datasets are also used to place billboards effectively. Consider for any city a trajectory database $\mathcal{D}$ is available. This database contains the location information of people along with the corresponding time stamps. Locations can be of different kinds e.g.; `Mall', `Beach', `Metro Station' and so on. Digital billboards are placed over those places and these billboards can be hired by different E-Commerce houses to show their advertisement content. Now, due to financial constraints, only very limited number of billboard slots can be hired. It's a natural research question among the available $n$ slots which $k$ of them (where $k << n$) should be chosen such that the influence is maximized. This is the problem that we have worked out in this paper.

\paragraph{\textbf{Related Work}} There are several studies in the literature related to  billboard advertisements. Zahradka et al. \cite{zahradka2021price} did a case study analysis of the cost of billboard advertising in the different regions of the Czech Republic. Zhang et al. \cite{zhang2020towards} studied the trajectory\mbox{-}driven influential billboard placement problem where a set of billboards along with its location and cost is given. The goal here is to choose a subset of the billboards within the budget that influence the largest number of trajectories. Wang et al. \cite{wang2022data} studied the problem of the  Targeted Outdoor Advertising Recommendation (TOAR) problem considering user profiles and advertisement topics. Their main contribution was a targeted influence model that characterizes the advertising influence spread along with the user mobility. Based on the divide and conquer approach they developed two solution strategies. Implementation with real-world datasets shows the efficiency and effectiveness of the proposed solution approaches. Also, there are few studies in the context of the billboard advertisement, that consider the minimization of regret that causes due to providing influence to the advertiser. Zhang et al. \cite{zhang2021minimizing} studied the problem of regret minimization and proposed several solution methodologies. Experimentation with real-world datasets showed the efficiency of the approaches.

 \paragraph{\textbf{Our Contributions}} In all the existing studies, the problem that has been considered is to identify influential locations for placing billboards. However, as now-a-days the billboards are digital, and they are allocated slot wise. So, it is important to consider these issues. In this paper we have made the following contributions: 
\begin{itemize}
\item We formulate the Influential Billboard Slot Selection Problem as a discrete optimization problem and showed that this problem is NP-Hard and hard to approximate within a constant factor.
\item We  propose a Pruned Submodularity Graph-based solution approach to solve this problem with its detailed analysis and illustration with a problem instance.
\item We conduct an extensive set of experiments with real-world billboard and trajectory datasets and compare the performance of the proposed algorithm with the existing solution approaches.   
\end{itemize}
\paragraph{\textbf{Organization of the Paper}}
The rest of the paper is organized as follows. Section \ref{Sec:PPD} describes the background and defines the problem formally. Section \ref{Sec:PA} describes the proposed solution approach. The experimental evaluations of the proposed solution approach have been described in Section \ref{Sec:EE}. Section \ref{Sec:CFD} concludes our study and gives future research directions.
\section{Preliminaries and Problem Definition} \label{Sec:PPD}
In this section, we describe the background and define the problem formally. Consider there are $m$ billboards $\mathcal{B}=\{b_1, b_2, \ldots, b_m\}$ and each one of them is running for the interval $[T_{1}, T_{2}]$ and assume that $T=T_{2}-T_{1}$. Also, assume that all the billboards are allocated slot-wise for display advertisement, and the duration of each slot is fixed and it is denoted by $\Delta$. These billboards are placed at different locations (e.g., street junctions, shopping malls, airports. metro stations, etc.) of a city. If some person $u_j$ comes close to any billboard $b_i$ at time $t$ and at that time if the advertisement content for some commercial house is running on that billboard then $u_j$ will be influenced towards the item with the probability $Pr(b_i,u_j)$. In this study, we assume that this value is known. However, the standard way of computing these values has also been described in the literature \cite{zhang2020towards}. For any positive integer $n$, by $[n]$ we denote the set $\{1,2, \ldots,n\}$. We denote any arbitrary billboard slot as a tuple containing two items: the first one is billboard id and the second one is the starting and ending time of a slot.  Let, $\mathbb{BS}$ denotes the set of all billboard slots; i.e.; $\mathbb{BS}=\{(b_i,[t_j,t_j+\Delta]): i \in [k] \text{ and } j \in \{1, \Delta+1, 2 \Delta +1, \ldots, \frac{T}{\Delta}+1\} \}$. Now, if $|\mathbb{BS}|=n $ then it is easy to observe that $n =k \cdot \frac{T}{\Delta}$.

 A trajectory database contains the location information of different persons and this is stated in Definition \ref{Def:TD}.
\begin{definition} [Trajectory Database] \label{Def:TD}
A trajectory database $\mathcal{D}$ is a collection of tuples of the following form :$<u_{id}, \texttt{loc}, \texttt{time\_stamp}>$. The description of each attribute are given below:

\begin{itemize}
\item $u_{id}$: This is the unique identification of a people.
\item \texttt{loc}: This is the location information of the people $u_{id}$.
\item \texttt{time}: This attribute stores the time information 
\end{itemize}
If there is a tuple $<u_{126}, \texttt{Chihago\_Airport}, [1800,2000]>$ in the trajectory database $\mathcal{D}$ then it signifies that the people with its unique identification number $u_{126}$ was at the place \texttt{Chicago\_Airport} from time stamp $1800$ to $2000$. 
\end{definition}
 One important point to highlight here is that in Definition \ref{Def:TD} we have listed  only the required attributes. However, in real-world datasets, we may have some more attributes as well; such as \texttt{trip\_id}, \texttt{vehicle\_id}, and many more. The set of unique user\_ids that are present in the trajectory database $\mathcal{D}$ is $\mathcal{U}=\{u_1, u_2, \ldots, u_n\}$. Next, we describe the Billboard Database.
 \begin{definition}[Billboard Database]
 The billboard database $\mathbb{B}$ contains the tuples of the following form: $<b_{id}, \texttt{loc}, \texttt{cost}>$. The meaning of each attribute has been given below:
 \begin{itemize}
 \item $b_{id}$: This attribute stores the unique ids of the billboards.
 \item \texttt{loc}: This attribute stores the location information of the billboard.
 \item  \texttt{cost}: This attribute stores the cost that needs to be paid by the E-Commerce house for renting one billboard. 
 \end{itemize}
 If there is a tuple $<b_{245}, \texttt{Chihago\_Airport}, 6>$ in the billboard database $\mathbb{B}$, then it signifies that the billboard whose unique id is $b_{245}$ is placed at the location \texttt{Chicago\_Airport} and $ \$ 6$ needs to be paid by the E-Commerce house for renting one billboard slot.
 \end{definition}
  Now, assume that the people $u_i \in \mathcal{U}$, is at the location \texttt{Chicago\_Airport} for the duration $[t_i, t^{'}_i]$ and for the duration $[t_j, t^{'}_j]$ at the billboard $b_i$ an advertisement of the brand XYZ. If $[t_j, t^{'}_j] \ \bigcap \ [t_i, t^{'}_i] \neq \emptyset$, then we can hope that the people $u_i$ will look into the advertisement and he will be influenced with the probability $Pr(b_j,u_i)$. In this study, we assume that for all $b_j \in \mathbb{BS}$, and people $u_i \in  \mathcal{U}$, we have these probability values. Next, we state the notion of the influence for a given subset of billboard slots in Definition \ref{Def:3}.
 
 \begin{definition} [Influence of Billboard Slots] \label{Def:3}
 Given a subset of billboard slots $\mathcal{S} \subseteq \mathbb{BS}$, we denote its influence by $I(\mathcal{S})$ and defined it as the sum of the influence probabilities of the individual people. Mathematically, this is characterized by Equation \ref{Eq:1}.
 \begin{equation} \label{Eq:1}
 I(\mathcal{S})= \underset{u_i \in \mathcal{U}}{\sum} [1 \ - \ \underset{b_j \in \mathbb{BS}}{\prod} (1-Pr(b_j,u_i))]
 \end{equation}
\end{definition} 
 It can be easily observed that the influence function $I$ is a set function that maps each possible subsets of billboard slots to their respective influence; i.e.; $I: 2^{\mathbb{BS}} \longrightarrow \mathbb{R}^{+}_{0}$ where $I(\emptyset)=0$. We list out the properties of the influence function in Lemma \ref{Lemma:1}. Due to space limitations, we are unable to provide its proof which will come in a subsequent journal version of this paper.
 
 \begin{lemma} \label{Lemma:1}
 The influence function $I()$ follows non-negativity, monotonicity, and submodularity property.
 
 \end{lemma}

  Now, it is important to observe that as the selection of billboard slots is involved with money, hence a limited number of them can be chosen. In commercial campaigns, the goal will be to maximize the influence. Then the question arises which $k$ billboard slots should be chosen to maximize the influence. This problem has been referred to as the \textsc{Influential Billboard Slot Selection} Problem which is stated in Definition \ref{Def:Problem}. 
  \begin{definition} [Influential Billboard Slot Selection Problem] \label{Def:Problem}
  Given a trajectory database $\mathcal{D}$ and its corresponding billboard database $\mathbb{B}$, the problem of Influential Billboard Slot Selection asks to choose a subset of $k$ billboard slots such that the influence as stated in Equation \ref{Eq:1} is maximized. Mathematically, this problem can be stated as follows:
  \begin{equation}
  \mathcal{S}^{OPT}= \underset{\mathcal{S} \ \subseteq \ \mathbb{BS}, | \mathcal{S}|=k}{argmax} \ I(\mathcal{S})
  \end{equation}
  \end{definition} 
  From the algorithmic point of view, the \textsc{Influential Billboard Slot Selection Problem} can be given as the following text box.
  
\begin{center}
\begin{tcolorbox}[title=\textsc{Influential Billboard Slot Selection Problem}, width=12.5cm]
\textbf{Input:} The Set of Billboard Slots $\mathbb{BS}$, The Influence Function $I()$, The Trajectory Database $\mathcal{D}$, The number of slots $k$.

\textbf{Problem:} Find out a  set $\mathcal{S} \subseteq \mathbb{BS}$ with $|\mathcal{S}|=k$ such that $I(\mathcal{S})$ is maximized.
\end{tcolorbox}
\end{center}  
  
We denote any arbitrary instance of the Influential Billboard Slot Selection Problem by $I=(\mathbb{BS}, \mathcal{S},k)$. First, we show that the Influential Billboard Slot Selection Problem is \textsf{NP-Hard} by a reduction from the Set Cover Problem.
  \begin{theorem}
  The Influential Billboard Slot Selection Problem is \textsf{NP-Hard}.
  \end{theorem}
    
\begin{proof} (Outline)
We prove this statement by a reduction from the Hitting Set Problem. We denote an arbitrary instance of the set cover problem by $I^{'}=(\mathcal{U}, \mathcal{X}, k^{'})$. Here, $\mathcal{U}$ is the ground set, $\mathcal{X}$ is the collection of the subsets over the ground set. The goal here is to choose $k^{'}$ many elements from the ground set $\mathcal{U}$ such that every subset in $\mathcal{X}$ contains at least one element from the chosen elements. It is well known that this problem is NP-Hard \cite{vazirani2001approximation}.
\par Now, we provide a polynomial time reduction from the Hitting Set Problem to the Influential Billboard Slot Selection Problem. Without loss of generality, we assume that the elements of the ground set are the subset of the set of natural numbers and starting from $1$. Also, for simplicity, we assume that there is only one slot in a billboard. Now, the  construction is as follows. For every $i \in \mathcal{U}$, we create one location with its id as $\ell_{i}$, one billboard with id $b_{i}$, and place the billboard at that location. For every subset $x \in \mathcal{X}$, we create one trajectory that contains the locations. We fix $k=k^{'}$. We want to influence all the trajectories. Now, it is easy to observe that the hitting set problem instance will have a solution of size $k^{'}$ if and only if the influential billboard slot selection problem instance has a solution of size $k$. Due to space limitations, we have only given an outline of the proof. 
\end{proof}

\section{Proposed Solution Approach}  \label{Sec:PA}
In this section, we describe the proposed solution approaches for this problem. Initially, we start by describing the Marginal Influence Gain of a Billboard Slot in Definition \ref{Def:MIG}.

\begin{definition} [Marginal Influence Gain of a Billboard Slot]
 \label{Def:MIG}
Given a subset of billboard slots $\mathcal{S} \ \subseteq \ \mathbb{BS}$ and a particular billboard $b \in \mathbb{BS} \setminus \mathcal{S}$, the marginal influence gain of the billboard slot $b$ with respect to the billboard slots in $\mathcal{S}$ is denoted by $\Delta (b| \mathcal{S})$ and defined as the difference between the influence when the billboard slot is included with $\mathcal{S}$ and when it is not. Mathematically, this is stated in Equation 2
\begin{equation}
\Delta (b \ | \ \mathcal{S}) = I(\mathcal{S} \ \cup \ \{b\}) \ - \ I(\mathcal{S})
\end{equation}
\end{definition}
 As shown that the influence function follows the summodularity property, hence this problem reduces to the problem of submodular function maximization subject to the cardinality constraint. This can be solved in the following way. Starting with an empty set, in each iteration we pick up the slot that causes maximum marginal gain. As shown by Nemhauser et al. \cite{nemhauser1978analysis,fisher1978analysis} this method leads to $(1-\frac{1}{e})$ factor approximation guarantee. However, it is easy to observe that in each iteration the number of marginal gain computation is of $\mathcal{O}(n)$ where $n$ is the total number of billboard slots which is quite large. This leads to huge computational burden for real-world problem instances. So, there is a need to develop efficient solution methodology. In this paper, we apply the pruned submodularity graph-based pruning technique to solve our problem \cite{zhou2017scaling}. Before describing the proposed solution methodology, first we introduce the notion of pruned submodularity graph in Definition \ref{Def:PS_Graph} which is the heart of the proposed solution approach.
 
\begin{definition}  [Pruned Submodularity Graph ] \label{Def:PS_Graph}
This is a weighted, directed graph $G(V,E,w)$ where the vertex set is the set of billboard slots; between every pair of billboard slots $b_i,b_j$, $(b_ib_j)$ will be an edge in $G$. The edge weight function $w$ maps each edge to a real number. Hence, the vertex set $V(G)=\mathbb{BS}$, the edge set $E(G)=\{(b_ib_j):\ i,j \in [p] \text{ and } i \neq j\}$, and the edge weight function $w$ is defined as follows:
\begin{equation} \label{Eq:weight}
w(b_i,b_j)= I(b_j \ \vert \ b_i) \ - \ I(b_i \ \vert \ \mathbb{BS} \setminus b_i)
\end{equation}
\end{definition}
This edge weight has got practical significance. Consider an edge $(b_ib_j)$ and its edge weight $w(b_ib_j)$. It is important to observe that the quantity $I(b_j \ \vert \ b_i)$ signifies that the maximum influence that the billboard slot can contribute when the billboard slot $b_i$ is already contained in the solution. Whereas the quantity $I(b_i \ \vert \ \mathbb{BS} \setminus b_i)$ signifies the billboard slot. As in our problem context, a vertex in the pruned submodularity graph corresponds to a billboard slot, hence in the rest of the paper we use the terminology `vertex' and `slot' interchangeably. Next, we state the divergence of a node of a pruned submodularity graph in Definition \ref{Def:Divergence}.

\begin{definition}  [Divergence of a Node] \label{Def:Divergence}
On the pruned submodularity graph $G(V,E,w)$, the divergence of a node $v$ from a set of nodes $V^{'}$ is defined as $w_{V^{'}v}= \underset{x \in V^{'}}{min} \ w_{xv}$.
\end{definition}

\par Next, we present the description of the proposed solution approach.

\paragraph{\textbf{Description of the Algorithm}} Now, we describe our proposed solution methodology. First, we perform the following preprocessing task. Billboard slots having their individual influence $0$ is removed from the list. As the influence function is submodular, for any billboard slot the marginal gain with respect to any set will always be less than its individual influence. After that, we construct the pruned submodularity graph with the slots having non-zero individual influence at Line No. $8$. Once we are done with the construction of the tree, we perform the pruning from Line No. $9$ to $16$. This works in the following way. In each iteration of the \texttt{while} loop we randomly sample out $r \cdot \log n$ many slots and put them in the list $\mathcal{U}$. Now, for all the remaining slots in $\mathbb{BS}$, we compute the value of $w_{\mathcal{U}v}$ as mentioned in Definition \ref{Def:Divergence}. From this list, we remove $(1-\frac{1}{\sqrt{c}})$ fraction of the slots having the smallest value of $w_{\mathcal{U}v}$. Once we are done with the pruning step, we execute the incremental greedy approach to finally select the billboard slots. The described solution approach has been represented in the form of pseudo code in Algorithm \ref{Algo: Algorithm1_Simple Greedy}.


 \begin{algorithm}[h]
\SetAlgoLined
\KwData{$\mathbb{BS}$,  $I()$, $\mathcal{D}$, $r$, $c$, and, $k$}
\KwResult{  A subset of $\mathbb{BS}$}
 Initialize $\mathcal{S} \leftarrow \emptyset$ \;
 \For{$\text{All }b \in \mathbb{BS}$}{
 \If{$I(b)==0$}{
 $\mathbb{BS} \longleftarrow \mathbb{BS} \setminus \{b\}$\;
 }
 }
 Initialize $\mathcal{S}^{'} \leftarrow \emptyset$, $\mathcal{U} \longleftarrow \emptyset$, $|\mathbb{BS}|=n$ \;
 $\text{Construct the pruned submodularity graph with the slots in }\mathbb{BS}$\;
 \While{$|\mathbb{BS}| > \ r \cdot \log n$}{
 $\text{Sample out }r \cdot \log n \text{ slots uniformly at random from }\mathbb{BS}  \text{ and place them in }\mathcal{U}$\; 
 $\mathbb{BS} \longleftarrow \mathbb{BS} \setminus \mathcal{U}$, $\mathcal{S}^{'} \longleftarrow \mathcal{S}^{'} \cup \mathcal{U} $\;
 \For{$\text{All }b \in \mathbb{BS} $}{
 $w_{\mathcal{U}b} \longleftarrow \ \underset{u \in \mathcal{U}}{min} \ [I(b|u)-I(u| \mathbb{BS} \setminus \{u\})]$\;
 }
 $\text{Remove }(1-\frac{1}{\sqrt{c}}) \cdot |\mathbb{BS}| \text{ elements from }\mathbb{BS} \text{ having smallest value of }w_{\mathcal{U}v}$\;
 }
 $\mathbb{BS} \longleftarrow \mathbb{BS} \  \bigcup \ \mathcal{S}^{'} $\;
 \While{$|\mathcal{S}| \neq k$}{
 $u^{*} \longleftarrow \underset{u \in \mathbb{BS} \setminus \mathcal{S}}{argmax} \ I(\mathcal{S} \cup \{u\}) - I(\mathcal{S})$\;
 $\mathcal{S} \longleftarrow \mathcal{S} \cup \{u^{*}\}$\;
 }
Return $\mathcal{S}$\;
 \caption{ Pruned Submodularity Graph+Incremental Greedy Approach for Influential Billboard Slot Selection Problem}
 \label{Algo: Algorithm1_Simple Greedy}
\end{algorithm}	
\paragraph{\textbf{Complexity Analysis}} Now, we analyze our proposed solution approach to understand its time and space requirements. Initialization at Line No. $1$ will take $\mathcal{O}(1)$ time. Now, for any arbitrary billboard slot $b \in \mathbb{BS}$, its individual influence computation using Equation No. \ref{Eq:1} will take $\mathcal{O}(t)$ time where $t$ is the number of tuples in the trajectory database. As there are $n$ billboard slots, the time required for execution from Line No. $2 $ to $6$ will take $\mathcal{O}(n \cdot t)$ time. After that, the initialization statements of Line No. $7$ take $\mathcal{O}(1)$ time. In the worst case, it may so happen that all the billboard slots lead to non-zero influence. In that case, the pruned submodularity graph will have $n$ vertices. So, there are $\mathcal{O}(n^{2})$ billboard slot pairs for whom the edge weight needs to be calculated. From Equation No. \ref{Eq:weight}, it is easy to observe that for each edge, its corresponding weight can be calculated in $\mathcal{O}(n \cdot t)$ time. Hence, the construction of the pruned submodularity graph will take $\mathcal{O}(n^{2} \cdot t)$ time. It is easy to observe that the number of iterations of the \texttt{while loop} will be of $\mathcal{O}(\log_{\sqrt{c}} n)$. Removing $\mathcal{O}(\log n)$ many elements from $\mathbb{BS}$ uniformly at random will take $\mathcal{O}(\log n)$ time. Also, in each iteration, the cost of computing the quantity $w_{\mathcal{U}v}$ is also involved. In the set $\mathcal{U}$, there are $\mathcal{O}(r \cdot \log n)=\mathcal{O}(\log n)$ many slots. In the pruned submodularity graph every vertex is directly linked with every other vertex in the graph. Now, for every vertex $b \in \mathbb{BS}$ in the pruned submodularity graph, we need to compute the function as mentioned in Line No. $13$ of Algorithm \ref{Algo: Algorithm1_Simple Greedy}. Now, we can store these values while computing the weights of the edges of the pruned submodularity graph. So, for a billboard slot $b \in \mathbb{BS}$ and any $u \in \mathcal{U}$ the computation $w_{bu}$ can be computed in $\mathcal{O}(1)$ time. As there are $\mathcal{O}(\log n)$ many elements in $\mathcal{U}$, so computing the value of $w_{\mathcal{U}b}$ will take $\mathcal{O}(\log n)$. As the number of elements in $\mathbb{BS}$ is upper bounded by $\mathcal{O}(n)$, performing the steps from Line No. $12$ to $14$ will take $\mathcal{O}(n \log n)$ time. Sorting these values will take $\mathcal{O}(n \log n)$ time and removing $(1-\frac{1}{\sqrt{c}})$ fraction of the elements of $\mathbb{BS}$ will take $\mathcal{O}(n)$ time in the worst case. So, a single iteration of the \texttt{while} loop will take $\mathcal{O}(\log n + n \cdot \log n + n)=\mathcal{O}(n \cdot \log n)$ time. So, the execution time of the \texttt{while} loop will be of $\mathcal{O}(n \cdot \log n \cdot \log_{\sqrt{c}} n)=\mathcal{O}(n \cdot \log^{2} n)$. Now, it is easy to observe that the \texttt{while} loop at Line No. $18$ will iterate $\mathcal{O}(k)$ many times. The only operation involved within this \texttt{while} loop is the marginal gain computation. For any billboard slot $b \in \mathbb{BS} \setminus \mathcal{S}$, its marginal gain computation will take $\mathcal{O}(n \cdot t)$ time. So, the execution of this \texttt{while} loop will take $\mathcal{O}(k \cdot n \cdot t)$ time. Hence, the total time requirement of Algorithm \ref{Algo: Algorithm1_Simple Greedy} will be of $\mathcal{O}(n \cdot t + n^{2} \cdot t + n \cdot \log^{2} n + k \cdot n \cdot t)$. As $k < < n$, this quantity is reduced to $\mathcal{O}( n^{2} \cdot t + n \cdot \log^{2} n)$.
\par Additional space requirement by Algorithm \ref{Algo: Algorithm1_Simple Greedy} is to store the following lists $\mathcal{S}$, $\mathcal{S}^{'}$ and $\mathcal{U}$ which will take $\mathcal{O}(n)$, $\mathcal{O}(n)$, and $\mathcal{O}(\log n)$, respectively. Other than this storing the pruned submodularity graph will take $\mathcal{O}(n^{2})$ space. One important thing to observe here is that we can store the edge weights in the adjacency matrix of the pruned submodularity graph and this can be used for the computation involved in Line No. $13$. This does not require any more extra space. So, the total space requirement by Algorithm \ref{Algo: Algorithm1_Simple Greedy} is of $\mathcal{O}(n + \log n +n^{2})=\mathcal{O}(n^{2})$. Hence, Theorem \ref{Th:Time_and_Space} holds.

\begin{theorem} \label{Th:Time_and_Space}
Time and space requirement by Algorithm \ref{Algo: Algorithm1_Simple Greedy} is of $\mathcal{O}( n^{2} \cdot t + n \cdot \log^{2} n)$ and $\mathcal{O}(n^{2})$, respectively.
\end{theorem}
 
\paragraph{\textbf{Theoretical Analysis Related to the Performance}} Due to the space limitation, we do not elaborate the analysis. In our analysis, we have used the result from the study in \cite{zhou2017scaling} and we state the result in Theorem \ref{Zohu_Result}.

\begin{theorem} \cite{zhou2017scaling} \label{Zohu_Result}
If we apply the pruned submodularity graph-based pruning technique then the size of the reduced ground set will be $|V^{'}|=\frac{cp}{\log \sqrt{c}} \cdot k \cdot \log^{2}n$ where $p$ is the probability of sampling. With high probability (e.g.; $n^{1-qp} \cdot \log_{\sqrt{c}}n$), for all $v \in V \setminus V^{'}$, $w_{V^{'}v} \leq 2 \cdot w_{V^{*}}$. Thus the incremental greedy algorithm for the submodular function maximization on the ground set leads to a solution $\mathcal{S}^{'}$ that follows the following criteria:
\begin{equation}
I(\mathcal{S}^{'}) \geq (1-\frac{1}{e}) \cdot (I(\mathcal{S}^{OPT})- 2k\epsilon)
\end{equation}
where $\mathcal{S}^{OPT}$ is the set of optimal slots of size $k$. 
\end{theorem}

\par One important thing to point out is that in Algorithm \ref{Algo: Algorithm1_Simple Greedy} there are two parameters $r$ and $c$. The parameter $c$ controls the shrink rate of the billboard slots (decreases at a rate of $\frac{1}{\sqrt{c}}$) whereas the parameter $r$ controls the size of the set $\mathcal{U}$, and eventually it controls the size of the final output of the pruning method. Now, it is important what value we should choose for $c$ and $r$. As mentioned in the study by Zhou et al. \cite{zhou2017scaling}, when the value of both $r$ and $c$ is $8$, the pruning method converges very quickly. In this study also we consider these values only. 

\section{Experimental Evaluation} \label{Sec:EE}
In this section, we describe the experimental evaluation of the proposed solution methodology. Initially, we start by describing the datasets.
\subsection{Datasets Used}
The datasets that have been used in our study have also been used by existing studies as well \cite{zhang2020towards}. The trajectory data is obtained from two real-world datasets. The first one is the TLC trip record dataset \footnote{\url{http://www.nyc.gov/html/tlc/html/about/trip_record_data.shtml.}} for NYC and the second one is the Foursquare check-in dataset \footnote{\url{https://sites.google.com/site/yangdingqi/home.}} Both these datasets contain the records of green taxi trips from Jan 2013 to Sep 2016 and different location types such as Mall, Beach, Airport, and so on. From these trajectory datasets, we separate the tuples corresponding to the following six locations: `Beach', `Mall', `Bank', `Bus Stop', `Train Station', `Airport', and create six different datasets. The basic statistics of this dataset have been listed in Table \ref{table:Dataset}.
\begin{table} [h!]
    \centering
    \caption{Basis Statistics of the Datasets}
    \begin{tabular}{ || c c c c || }
    \hline
    Location Type & \# Rows1 & \# Rows2  & \# Billboard Slots\\ [0.5 ex]
    \hline \hline
    Beach & 76 & 575 & 21888 \\
    Mall & 86 & 1186 & 24768 \\
    Bank & 671 & 2232 &  193248 \\
    Bus Stop & 1056& 4473 &  304128 \\
    Train Station & 288 & 6407 & 82944 \\
    Airport & 313 & 2852 & 90144 \\ [1ex]
    \hline
    \end{tabular}
    \label {table:Dataset}
    \end{table}
\subsection{Experimental Setup}
Now, we describe the experimental setup involved in our experiments. The only parameter whose value needs to be set in this study is the influence probability whose value needs to be fixed. In the dataset (provided by LAMAR \footnote{\url{http://www.lamar.com/InventoryBrowser.}}) we have information about the panel size of the billboards. Among all the billboards, we choose the panel size of the biggest billboard and we fix the influence probability as the ratio between the size of that billboard and the size of the biggest billboard. After fixing the influence probabilities, we compute the influence for four different values of $k$, namely $10$, $15$, $20$, $25$.  
\subsection{Algorithms Compared}
We compare the performance of the proposed solution approach with the following approaches:
\begin{itemize}
\item \textbf{RANDOM}: In this method, for a given value of $k$, we pick any $k$ billboard slots uniformly at random. This is possibly the most simple approach and it will take $\mathcal{O}(k)$ time to execute.
\item \textbf{Top-$k$}: In this method, for every individual billboard slot we compute their influence and sort them in descending order based on their individual influence value. From this sorted list, we pick up Top-$k$ slots. If there are $s$ many tuples in the trajectory dataset, then computing influence for a single billboard slot will take $\mathcal{O}(s)$ time. As there are total $n$ many billboard slots, hence the total time requirement for influence computation is of $\mathcal{O}(n \cdot s)$. Sorting the slots based on the individual influence value will take $\mathcal{O}(n \cdot \log n)$ time. Choosing $k$ of them will take additional $\mathcal{O}(k)$ time. Hence the total time required will be of $\mathcal{O}(n \cdot s + n \cdot \log n + k)$. As $k << n$, this has been reduced to $\mathcal{O}(n \cdot s + n \cdot \log n)$.
\item \textbf{Maximum Coverage (MAX\_COV)}: In this method, we calculate the individual coverage of every billboard slot. We say that a billboard slot $b \in \mathbb{BS}$ covers a tuple $t$ in the trajectory database $\mathcal{D}$ if the recorded time present in the tuple contained in the duration of the billboard slot. After computing the coverage for every billboard slot, we sort them in descending order and pick $k$ slots of them with the highest coverage value. For a single billboard slot, computing its coverage will take $\mathcal{O}(s)$ time. As there are $n$ many billboard slots hence computing coverage for all the slots will take $\mathcal{O}(n \cdot s)$ time. Now, sorting the slots based on the coverage value will take $\mathcal{O}(n \cdot \log n)$ time. Now, choosing $k$ slots from the sorted list will take $\mathcal{O}(k)$ time. As $k << n$, hence the total time requirement of this approach will be $\mathcal{O}(n \cdot s + n \cdot \log n)$.
\item \textbf{Pruned Submodularity Graph+Random (P.S.Graph+RAND)}: In this method, first the pruned submodularity graph-based approach is applied to obtain the pruned set of billboard slots, and from these slots we pick $k$ of them uniformly at random. 
\end{itemize} 
\subsection{Goals of the Experiments} \label{Sec:Goals}
The goals of the experiments are four folded and they are mentioned below:
\begin{itemize}
\item \textbf{The Effectiveness of Preprocessing}: The first experimental goal of our study is to understand how effective the preprocessing step is. As mentioned previously, if we remove the billboard slots having zero influence then due to submodularity property of the influence function, it does not make any difference.
\item \textbf{The Effectiveness of Pruning}: One of our goals is to study what is the percentage of the billboard slots pruned out and it is good if it is sufficiently large. In that case, we can easily apply the incremental greedy approach without much computational burden.  
\item \textbf{The Influence Spread}: Another goal of our study is to make a comparative study of the proposed as well as the baseline methods based on the influence spread.
\item \textbf{Computational Time Requirement}: Finally, our goal is to have a comparison regarding the computational time among the proposed as well as baseline methods. 
\end{itemize}
Next, we proceed to describe the experimental observations along with their explanations.
\subsection{Observations with Explanation}
\paragraph{\textbf{The Effectiveness of Preprocessing}} In most of the instances, we observe that a significant number of billboard slots are removed after the preprocessing phase. When the location type is `Beach', among the $21888$ many billboard slots only $98$ of them are remaining after the preprocessing phase. Percentage wise more than $99 \%$ of billboard slots are removed and these observations are consistent with the other location type as well. Table \ref{table:Preprocessing} lists out the number and percentage wise billboard slots removed after preprocessing for different location types. Here, $\mathbb{BS}^{'}$ denotes the set of billboard slots after preprocessing.
\begin{table} [h!]
    \centering
    \caption{Experimental results related to the Preprocessing Step}
    \begin{tabular}{ || c c c c || }
    \hline
    Location Type & $|\mathbb{BS}|$ & $|\mathbb{BS}^{'}|$  & Percentage\\ [0.5 ex]
    \hline \hline
    Beach & 21888 & 98 & 99.55 \% \\
    Mall & 24768 & 155 & 99.37 \% \\
    Bank & 193248 & 273 &  99.85 \% \\
    Bus Stop & 304128 & 1030 &  99.66 \% \\
    Train Station & 82944 & 804 & 99.03 \% \\
    Airport & 90144 & 557 & 99.38 \% \\ [1ex]
    \hline
    \end{tabular}
    \label {table:Preprocessing}
    \end{table}
\paragraph{\textbf{The Effectiveness of Pruning}} As mentioned previously, pruning is an important step in the proposed methodology. When the location type is `Beach', the number of billboard slots available after the preprocessing step is $98$. If we apply the pruned submodularity graph-based approach the number of slots is even reduced to $58$ and the percentage-wise reduction is more than $40 \%$. So, including both preprocessing and  pruning, the proposed solution approach reduces $99.997 \%$ slots. Even for other location types also our observation is consistent that the proposed preprocessing and pruning technique reduces a significant portion of the billboard slots. Table \ref{table:Pruning} contains experimental results related to the pruning method. Here, $\mathbb{BS}^{''}$ denotes the set of billboard slots after pruning. It is easy to observe that $\mathbb{BS}^{''} \subseteq \mathbb{BS}^{'} \subseteq \mathbb{BS}$.

\begin{table} [h!]
    \centering
    \caption{Experimental results related to the Pruning Step}
    \begin{tabular}{ || c c c c c|| }
    \hline
    Location Type & $|\mathbb{BS}^{'}|$ & $|\mathbb{BS}^{''}|$  & Pruning Percentage & (Preprocessing+Pruning) Percentage\\ [0.5 ex]
    \hline \hline
    Beach & 98 & 58 & 40.81 \% & 99.73 \% \\
    Mall & 155 & 73 & 52.9 \%  & 99.7 \% \\
    Bank & 273 & 95 & 65.2 \% & 99.95 \% \\
    Bus Stop & 1030 & 162 &  84.27 \% & 99.94 \% \\
    Train Station & 804 & 147 & 81.71 \% & 99.82 \% \\
    Airport & 557 & 128 & 77 \% & 99.85 \% \\ [1ex]
    \hline
    \end{tabular}
    \label {table:Pruning}
    \end{table}
    \vspace*{-0.5 cm}
\paragraph{\textbf{The Influence Spread}} Now, we describe our comparative study on influence spread. Figure \ref{Fig:Influence} shows the number of billboard slots vs. influence spread plots for different location types. From the figures we observe that the billboard slots selected by the proposed solution approach leads to the maximum influence compared to the baseline methods. As an example, when the value of $k$ is $25$, the influence due to the billboard slots selected by the pruned submodularity graph along with the incremental greedy approach is around $1169$. However, among the baseline methods the Top-$k$ method leads to maximum amount of influence and that is $655$. So, in this dataset the proposed solution approach leads to approximately $78 \%$ more influence compared to the baseline methods. This observations are consistent with other location types as well with few exceptions when the location type is `Mall' and the budgets are $10$ and $15$. In these two cases we observe that the Top-$k$ Approach is giving more influence. 
\begin{figure*}[!ht]
\centering
\begin{tabular}{cc}
\includegraphics[scale=0.2]{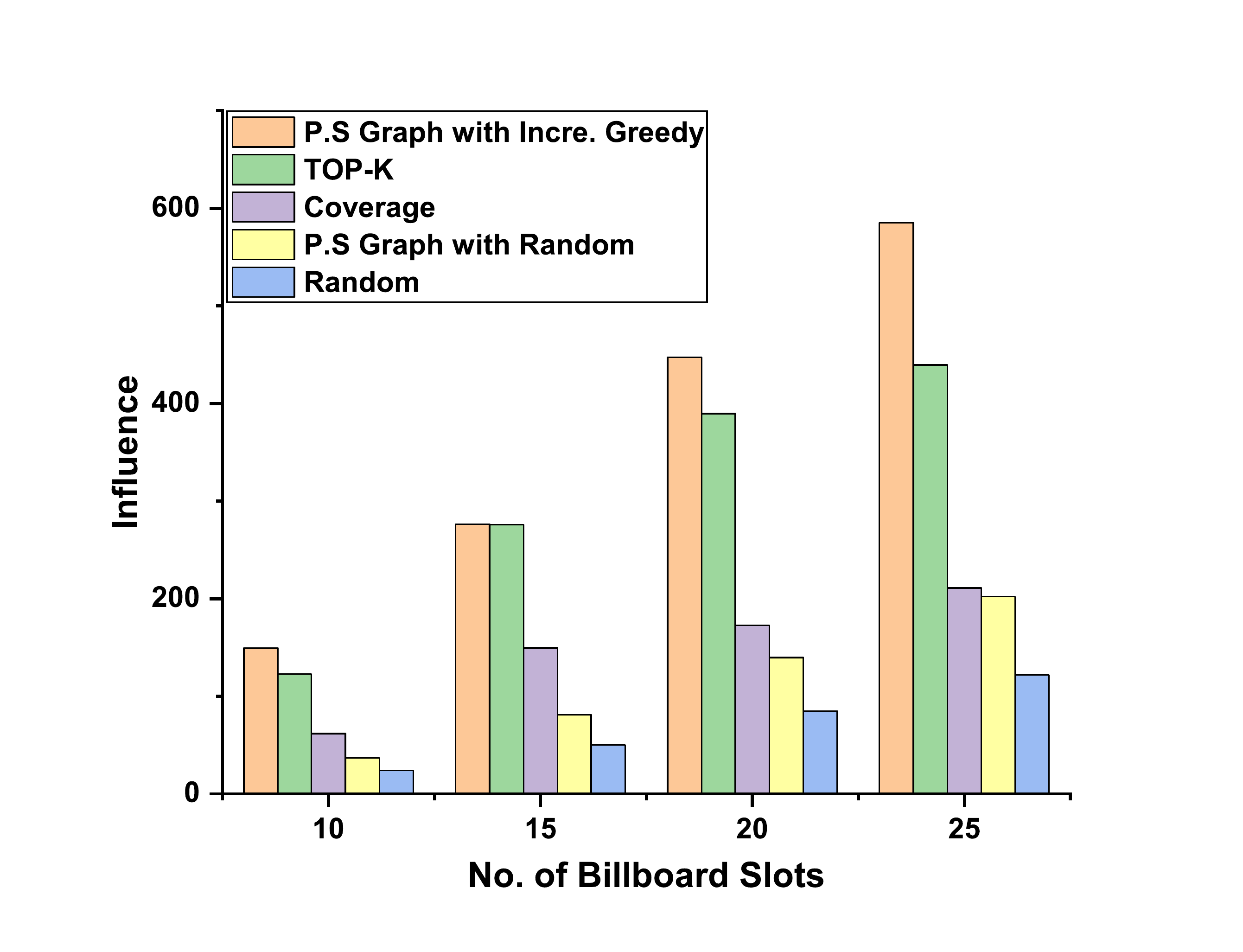} & \includegraphics[scale=0.2]{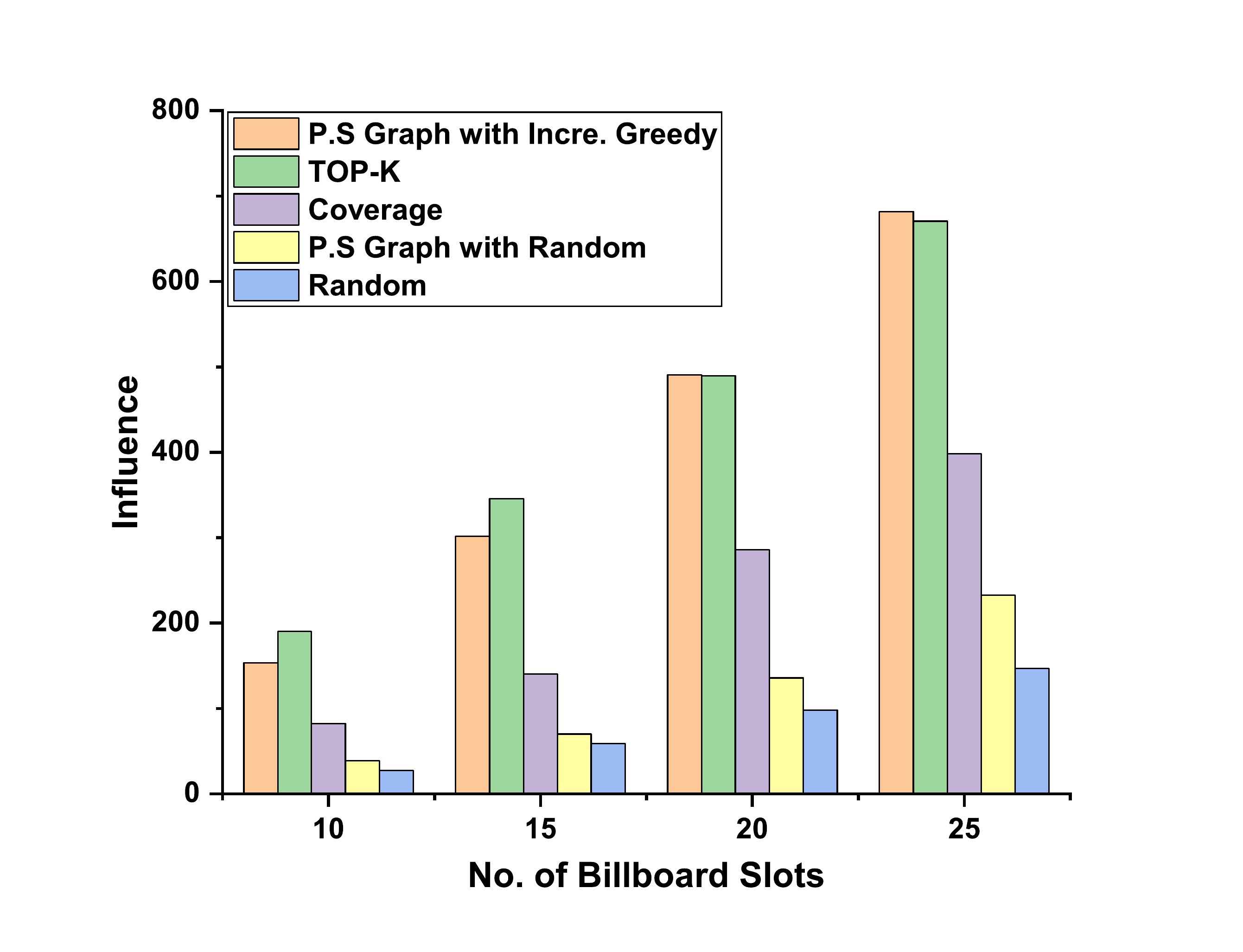}  \\
(a) Location Type=`Beach' & (b) Location Type=`Mall'  \\
\includegraphics[scale=0.2]{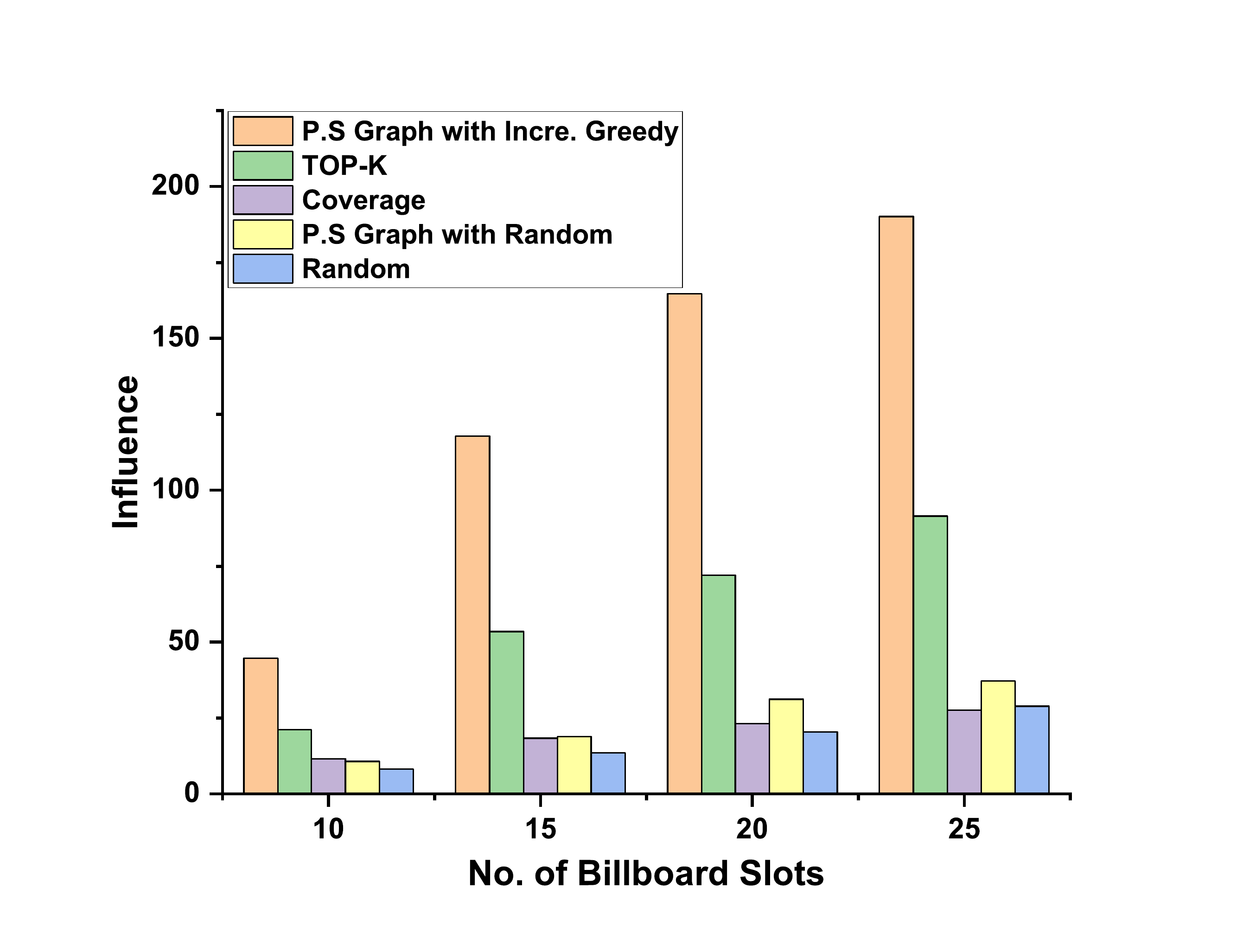} & \includegraphics[scale=0.2]{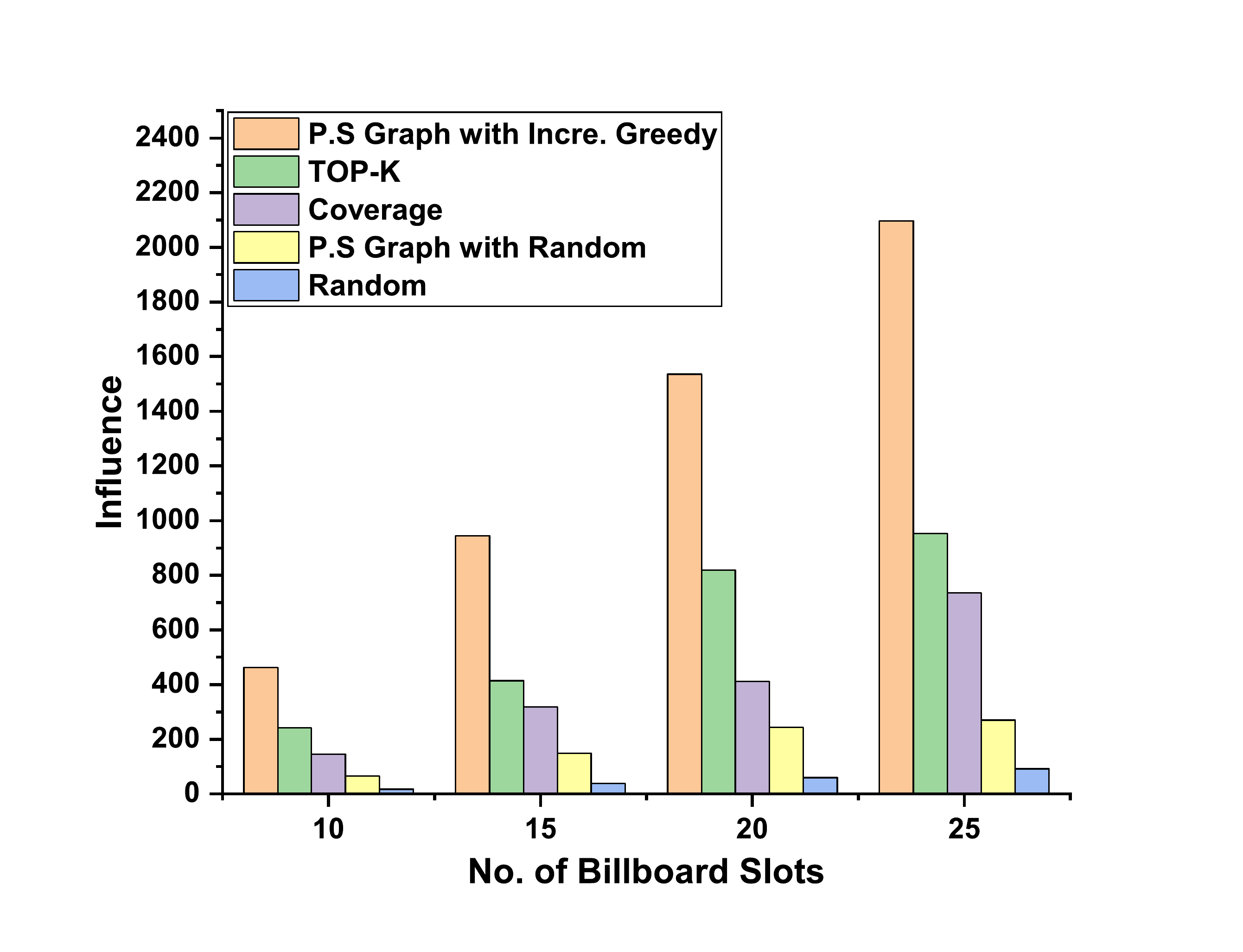}  \\
(c) Location Type=`Bank' & (d) Location Type=`Bus Stop'  \\
\includegraphics[scale=0.2]{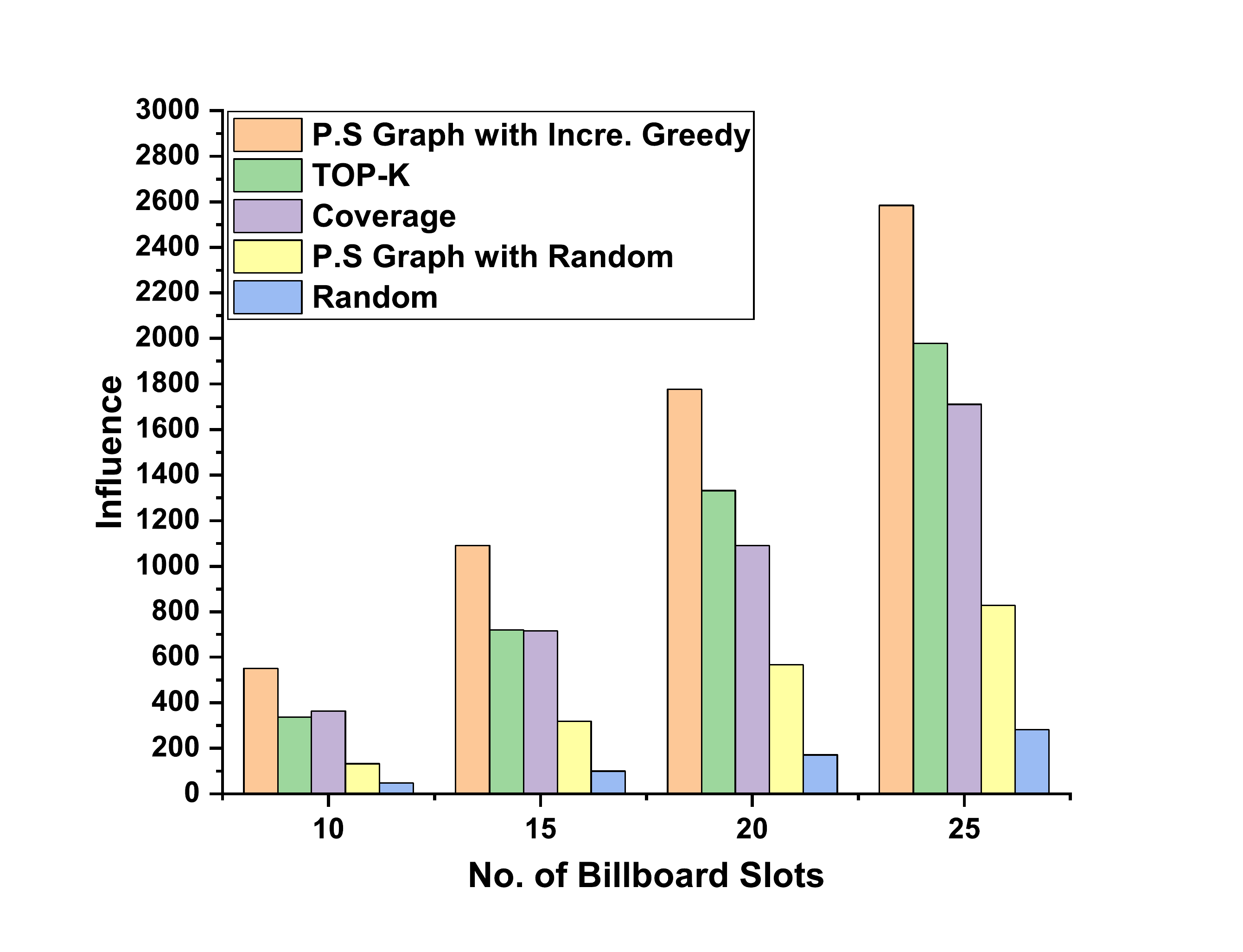} & \includegraphics[scale=0.2]{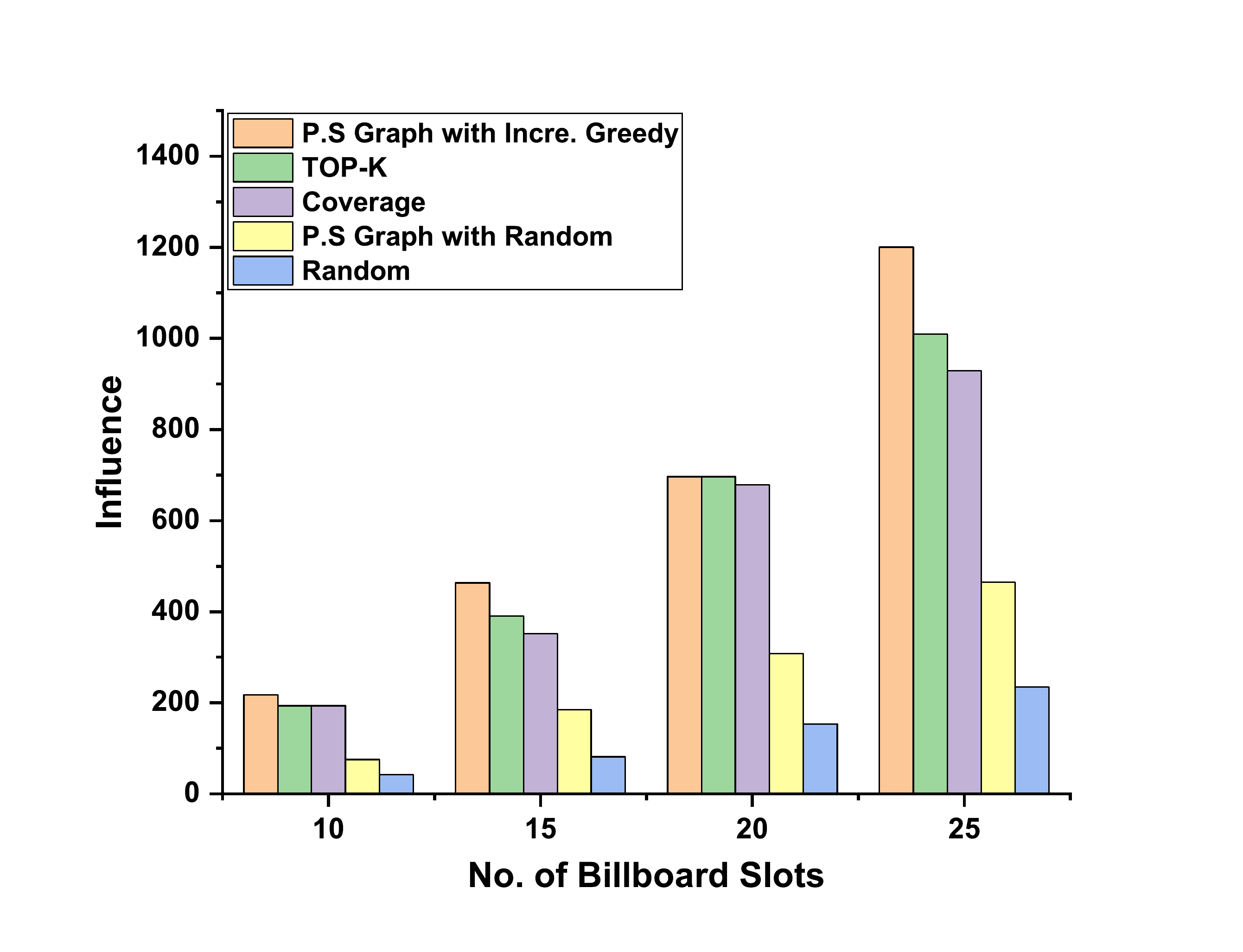}  \\
(e) Location Type=`Train Station' & (f) Location Type=`Airport'  \\
\end{tabular}
\caption{No. of Billboard Slots Vs. Influence plots for all the location types}
\label{Fig:Influence}
\end{figure*}
\paragraph{\textbf{Computational Time Requirement}} From the results, we observe that the computational time requirement of the proposed methodology is reasonable. As an example when the location type is `Beach' the preprocessing and the pruning  time requirement is $84$ Sec  and $110$ Sec, respectively. After that to apply the incremental greedy approach for the final selection of billboard slots requires $7.45$ minutes when the value of $k$ is $10$. With the increment of $k$, the time required for the incremental greedy approach is increasing rapidly. These observations are consistent with the other location types also. Table \ref{table:Computational_Time} contains the experimental results related to the computational time requirement. Due to the space limitations, we only provide the computational time required for $k=10$. In case of the PS Graph+Greedy Approach time requirement is $x+y+z$ means, preprocessing time is $x$ Sec, pruning time is $y$ Secs., and the selection time is $z$ Secs. 
\begin{table} [h!]
    \centering
    \caption{Computational time requirement for all the algorithms for different location types (in Secs.)}
    \begin{tabular}{ || c c c c c c|| }
    \hline
    Location Type & PS Graph+Greedy & Top-$k$  & Coverage & PS Graph+Random & Random \\ [0.5 ex]
    \hline \hline
    Beach & 84+110+447 & 64 & 105 & 190 & 120\\
    Mall & 191+487+1554 & 102 & 209 & 864 & 432\\
    Bank & 5527+4136+14004 & 5544 & 5544 & 10440 & 6372\\
    Bus Stop & 17174+42984+74160 & 17316 & 17352 & 61200 & 10440 \\
    Train Station &  5497+37605+57600 & 5652 & 5616 & 10440 & 7200\\
    Airport & 1625+7201+12600 & 1656 & 1656 & 9360 & 2196 \\ [1ex]
    \hline
    \end{tabular}
    \label {table:Computational_Time}
    \end{table}

\section{Conclusion and Future Research Directions} \label{Sec:CFD}
In this paper, we have studied the problem of influential billboard slot selection and posed this as a discrete optimization problem. We propose an effective solution approach that works in two phases. First, we apply the pruned submodularity graph-based approach to obtain the pruned set of billboard slots, and subsequently, we apply the incremental greedy approach for finally selecting the billboard slots. We analyze the proposed solution approach to understand its time and space requirements. Also, we have analyzed the proposed methodology to obtain the performance guarantee. Our future work on this problem will remain concentrated on developing efficient solution methodologies for this problem.
%

%
%
\bibliographystyle{splncs03}
\bibliography{Paper}
\end{document}